\documentclass[conference,onecolumn]{IEEEtran}

\usepackage{amsmath,amssymb,amsfonts,graphicx}
\usepackage{latexsym}
\usepackage{amsthm,cite}

\newtheorem{theorem}{Theorem}
\newtheorem{lemma}[theorem]{Lemma}

\newtheorem{corollary}[theorem]{Corollary}

\begin{document}
\sloppy

\title{Symmetric Two-User Gaussian Interference Channel with Common Messages}

\author{
  \IEEEauthorblockN{Quan Geng}
  \IEEEauthorblockA{CSL and Dept. of ECE\\
    UIUC, IL 61801\\
    Email: geng5@illinois.edu}
  \and
  \IEEEauthorblockN{Tie Liu}
  \IEEEauthorblockA{Dept. of Electrical and Computer Engineering\\
	Texas A\&M University, TX 77843 \\
    Email: tieliu@tamu.edu}
}

\maketitle

\begin{abstract}
We consider symmetric two-user Gaussian interference channel with common messages. We derive an upper bound on the sum capacity, and show that the upper bound is tight in the low interference regime, where the optimal transmission scheme is to send no common messages and each receiver treats interference as noise. Our result shows that although the availability of common messages provides a  cooperation opportunity for transmitters, in the low interference regime the presence of common messages does not help increase the sum capacity.

\end{abstract}

\section{Introduction} \label{sec:intro}

Interference channel is a fundamentally important communication model in information theory \cite{Cover12}. While the exact capacity region of interference channel in the simplest setting with two transmitter-receiver pairs is still unknown in general, recent research efforts have significantly improve our understanding of the capacity region. In particular, \cite{Etkin08} characterizes the capacity region of two-user Gaussian interference channel within one bit.  The exact capacity region has also been derived in certain regimes, e.g. the strong interference regime \cite{Sato81}, and the low interference regime \cite{Sreekanth09, Shang09, Motahari09}, which show that in the low interference regime treating interference as noise is optimal and achieves the sum capacity.

In this paper, we consider the symmetric two-user Gaussian interference channel with common messages, where each transmitter wants to send a private message to its corresponding receiver and both transmitters also intend to send a common message to both receivers. We derive an upper bound on the sum capacity using a genie-aided method \cite{Etkin08, Sreekanth09}, and show that the upper bound is tight in the low interference regime, where the optimal transmission scheme is to send no common messages and each receiver treats interference as noise. Our result shows that although the availability of common messages provides a cooperation opportunity for transmitters, in the low interference regime the presence of common messages does not help increase the sum capacity.

\subsection{Organization}

This paper is organized as follows. We describe the channel model in Section \ref{sec:model}, and derive an upper bound on the sum capacity in Section \ref{sec:upperbound}. In Section \ref{sec:tightness}, we give a natural lower bound on the sum capacity, and show that the upper bound matches the lower bound in certain low interference regime. In Section \ref{sec:zerocommon}, we prove that in the low interference regime the availability of the common messages does not help increase the sum capacity and thus treating interference as noise is optimal. Section \ref{sec:conclusion} concludes this paper.

\section{System Model} \label{sec:model}
We consider a symmetric two-user Gaussian interference channel with common messages. The channel input-output relation is given by
\begin{eqnarray}
Y_1 &=& X_1+cX_2+Z_1\\
Y_2 &=& X_2+cX_2+Z_2
\end{eqnarray}
where $X_i$ is the signal sent by the $i$th transmitter, and $Y_i$ is the signal received by the $i$th receiver, $Z_k$ is $\mathcal{N}(0,1)$ for $k=1,2$, and $E[X_k^2] \leq P$ for $k=1,2$. Without loss of generality, we assume that $c \geq 0$.

There are a set of three independent messages $(W_0,W_1,W_2)$, where $W_0$ is available at both transmitters and intended for both receivers, $W_1$ is available at transmitter 1 only and intended for receiver 1 only, and $W_2$ is available at transmitter 2 only and intended for receiver 2 only. We use $R_i$ to denote the transmission rate for messages $W_i$, for $i=0,1,2.$

\begin{figure}[t]
\centering
\includegraphics[scale=0.7]{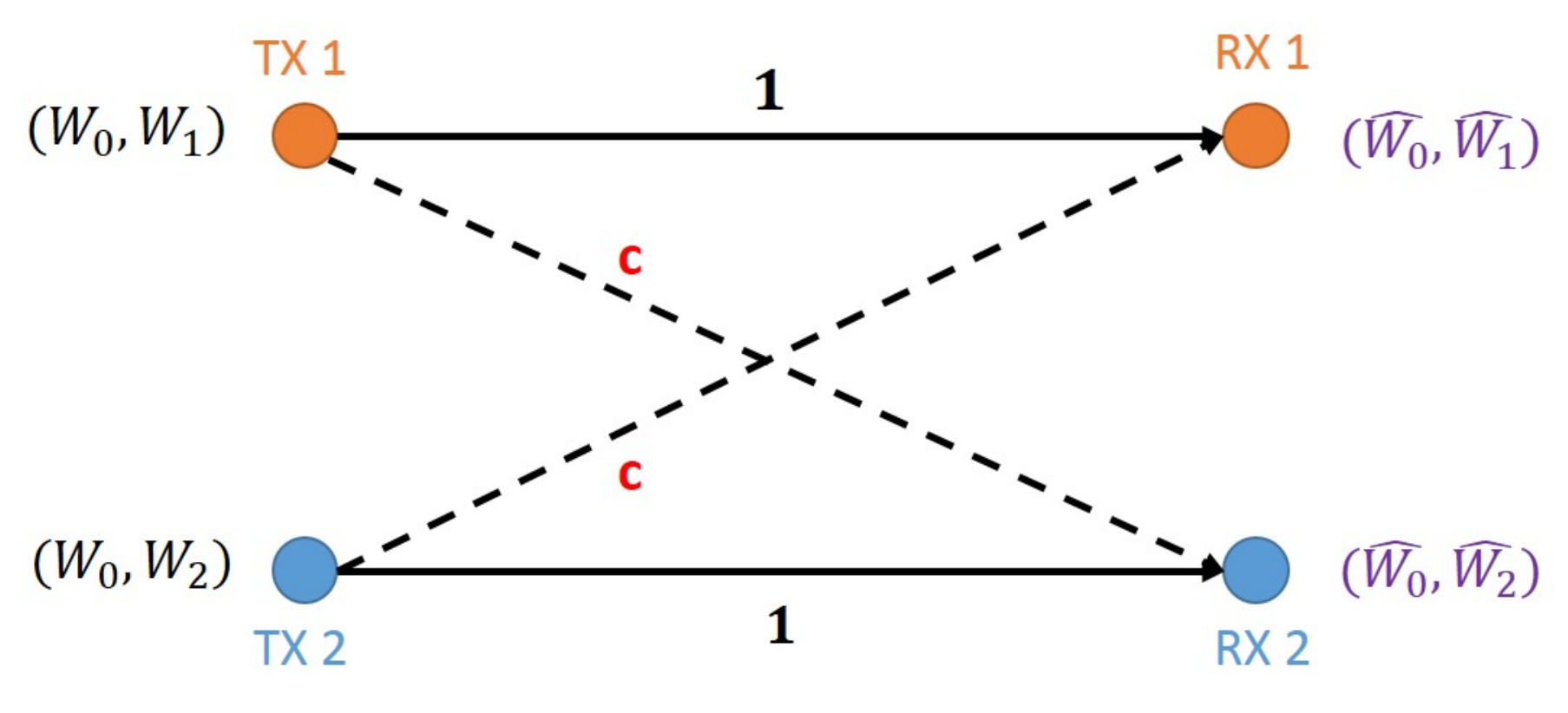}
\caption{Channel Model}
\label{fig:model}
\end{figure}

\section{Upper Bound for the Sum Capacity}\label{sec:upperbound}
In this section, we use the genie-added technique to derive an upper bound on the sum rates.

Our main result is given in Theorem \ref{thm:upper}.

\begin{theorem}\label{thm:upper}
The sum rate $R_0+R_1+R_2$ can be bounded from above as
\begin{align}
 R_0+R_1+R_2 \le \max_{0 \leq P_1 = P_2 \leq P} g(P_1,P_2),  \label{eqn:upper}
\end{align}
where
$g(P_1,P_2)$ is defined as
\begin{align}
 g(P_1,P_2) \triangleq \min_{ (a_1^2,a_2^2,\mathrm{Var}(\tilde{Z}_1),\mathrm{Var}(\tilde{Z}_2)) \in \mathcal{A}(P_1,P_2) } f\left(P_1,P_2,a_1^2,\mathrm{Var}(\tilde{Z}_1),a_2^2,\mathrm{Var}(\tilde{Z}_2)\right),
\end{align}
and
\begin{eqnarray}
\mathcal{A}(P_1,P_2) &:=&
\left\{\left.\left(a_1^2,a_2^2,\mathrm{Var}(\tilde{Z}_1),\mathrm{Var}(\tilde{Z}_2)\right)\right|\right.\nonumber\\
&& \left.\begin{array}{r}
  0 \leq \mathrm{Var}(\tilde{Z}_1) \leq 1-a_2^2\\
  \sqrt{\left(1-a_2^2-\mathrm{Var}(\tilde{Z}_1)\right)\left(1-\mathrm{Var}(\tilde{Z}_1)\right)}-\mathrm{Var}(\tilde{Z}_1) \geq
  c^2P_1\\
  0 \leq \mathrm{Var}(\tilde{Z}_2) \leq 1-a_1^2\\
  \sqrt{\left(1-a_1^2-\mathrm{Var}(\tilde{Z}_2)\right)\left(1-\mathrm{Var}(\tilde{Z}_2)\right)}-\mathrm{Var}(\tilde{Z}_2) \geq
  c^2P_2
\end{array}
\right\}
\end{eqnarray}
and
\begin{eqnarray}
&&
f\left(P_1,P_2,a_1^2,\mathrm{Var}(\tilde{Z}_1),a_2^2,\mathrm{Var}(\tilde{Z}_2)\right)\nonumber\\
&& \hspace{20pt} :=
\;\frac{1}{4}\left[\log\frac{\left(P+c^2P+2c\sqrt{(P-P_1)(P-P_2)}+1\right)^2}{(c^2P_1+1)(c^2P_2+1)}+\right.\nonumber\\
&& \hspace{35pt}
\log\frac{(P_1+c^2P_2+1)(c^2P_1+\mathrm{Var}(\tilde{Z}_1))-
\left(cP_1+a_1\sqrt{\mathrm{Var}(\tilde{Z}_1)}\right)^2
}{(c^2P_1+1-a_2^2)\mathrm{Var}(\tilde{Z}_1)}+\nonumber\\
&& \hspace{35pt}\left.
\log\frac{(P_2+c^2P_1+1)(c^2P_2+\mathrm{Var}(\tilde{Z}_2))-
\left(cP_2+a_2\sqrt{\mathrm{Var}(\tilde{Z}_2)}\right)^2
}{(c^2P_2+1-a_1^2)\mathrm{Var}(\tilde{Z}_2)}\right]
\end{eqnarray}
\end{theorem}

\begin{IEEEproof}
We first prove that
\begin{align}
  R_0+R_1+R_2 \le \max_{0 \leq P_1, P_2 \le P} g(P_1,P_2).
\end{align}

 From Fano's inequality, we have for $k=1,2$, any
$\epsilon>0$, and sufficiently large $n$
\begin{eqnarray}
n\left(R_0-\epsilon/3\right) & \leq & I(W_0;Y_k^n)\\
& = & h(Y_k^n)-h(Y_k^n|W_0)\\
& = & h(Y_k^n)-\left(h(Y_k^n|W_0)-h(Y_k^n|X_k^n,W_0)\right)-h(Y_k^n|X_k^n,W_0)\\
& = & h(Y_k^n)-I(X_k^n;Y_k^n|W_0)-h(Y_k^n|X_k^n,W_0)
\end{eqnarray}
Also have from Fano's inequality, we have for $k=1,2$, any
$\epsilon>0$, and sufficiently large $n$
\begin{eqnarray}
n\left(R_k-\epsilon/3\right) & \leq & I(W_k;Y_k^n)\\
& \leq & I(W_k;Y_k^n,W_0)\\
& = & I(W_k;Y_k^n|W_0) \label{eq:T1}\\
& \leq & I(X_k^n;Y_k^n|W_0)\label{eq:T2}
\end{eqnarray}
where \eqref{eq:T1} follows from the independence between $W_k$ and
$W_0$, and \eqref{eq:T2} follows from the fact that given $W_0$,
$W_k-X_k^n-Y_k^n$ forms a Markov chain. The sum rate
\begin{eqnarray}
&& n(R_0+R_1+R_2-\epsilon)\nonumber\\
&& \hspace{20pt} = \;
\frac{n}{2}(R_0-\epsilon/3)+\frac{n}{2}(R_0-\epsilon/3)+n(R_1-\epsilon/3)+n(R_2-\epsilon/3)\\
&& \hspace{20pt} \leq \;
\frac{1}{2}\left(h(Y_1^n)-I(X_1^n;Y_1^n|W_0)-h(Y_1^n|X_1^n,W_0)\right)+\nonumber\\
&& \hspace{35pt} \frac{1}{2}\left(h(Y_2^n)-I(X_2^n;Y_2^n|W_0)-h(Y_2^n|X_2^n,W_0)\right)+\nonumber\\
&& \hspace{35pt} I(X_1^n;Y_1^n|W_0)+I(X_2^n;Y_2^n|W_0)\\
&& \hspace{20pt} = \;
\frac{1}{2}\left(h(Y_1^n)+I(X_1^n;Y_1^n|W_0)-h(Y_1^n|X_1^n,W_0)\right.+\nonumber\\
&& \hspace{35pt}
\left.h(Y_2^n)+I(X_2^n;Y_2^n|W_0)-h(Y_2^n|X_2^n,W_0)\right)\\
&& \hspace{20pt} \leq \; \frac{1}{2}\left(h(Y_1^n)+I(X_1^n;Y_1^n,U_1^n|W_0)-h(Y_1^n|X_1^n,W_0)+\right.\nonumber\\
&& \hspace{35pt} \left.h(Y_2^n)+I(X_2^n;Y_2^n,U_2^n|W_0)-h(Y_2^n|X_2^n,W_0)\right)  \label{eqn:temp5}\\
&& \hspace{20pt} =\; \frac{1}{2}\left(h(Y_1^n)+I(X_1^n;U_1^n|W_0)+I(X_1^n;Y_1^n|U_1^n,W_0)-h(Y_1^n|X_1^n,W_0)+\right.\nonumber\\
&& \hspace{35pt} \left.h(Y_2^n)+I(X_2^n;U_2^n|W_0)+I(X_2^n;Y_2^n|U_2^n,W_0)-h(Y_2^n|X_2^n,W_0)\right)\\
&& \hspace{20pt} =\; \frac{1}{2}\left(h(Y_1^n)+h(U_1^n|W_0)-h(U_1^n|X_1^n,W_0)+\right.\nonumber\\
&& \hspace{35pt} h(Y_1^n|U_1^n,W_0)-h(Y_1^n|X_1^n,U_1^n,W_0)-h(Y_1^n|X_1^n,W_0)+\nonumber\\
&& \hspace{35pt} h(Y_2^n)+h(U_2^n|W_0)-h(U_2^n|X_2^n,W_0)+\nonumber\\
&& \hspace{35pt} \left.h(Y_2^n|U_2^n,W_0)-h(Y_2^n|X_2^n,U_2^n,W_0)-h(Y_2^n|X_2^n,W_0)\right)\\
&& \hspace{20pt} =\; \frac{1}{2}\left(-h(U_1^n|X_1^n,W_0)-h(U_2^n|X_2^n,W_0)+\right.\nonumber\\
&& \hspace{35pt} h(Y_1^n)+h(Y_2^n)+\nonumber\\
&& \hspace{35pt} h(Y_1^n|U_1^n,W_0)+h(Y_2^n|U_2^n,W_0)+\nonumber\\
&& \hspace{35pt} h(U_1^n|W_0)-h(Y_2^n|X_2^n,U_2^n,W_0)-h(Y_2^n|X_2^n,W_0)+\nonumber\\
&& \hspace{35pt}
\left.h(U_2^n|W_0)-h(Y_1^n|X_1^n,U_1^n,W_0)-h(Y_1^n|X_1^n,W_0)\right)
\label{eq:T3}
\end{eqnarray}
for any genie signals $(U_1^n,U_2^n)$.

Motivated by the problem of two-user Gaussian interference channel
\emph{without} common information, we shall choose
\begin{equation}
U_{ki} = cX_{ki}+\tilde{Z}_{ki}
\end{equation}
and $\tilde{Z}_{ki}$ are i.i.d.
$\mathcal{N}(0,\mathrm{Var}(\tilde{Z}_k))$ and are \emph{correlated}
with the noise signal $Z_{ki}$ as
\begin{equation}
Z_{ki}=\frac{a_k}{\sqrt{\mathrm{Var}(\tilde{Z}_k)}}\tilde{Z}_{ki}+N_{ki}
\end{equation}
Let
\begin{equation}
P_k:=\frac{1}{n}\sum_{i=1}^{n}\mathrm{Var}(X_{ki}|W_0)
\end{equation}
Note that
\begin{equation}
h(U_k^n|X_k^n,W_0)=h(cX_k^n+\tilde{Z}_k^n|X_k^n,W_0)=h(\tilde{Z}_k^n)
=\frac{n}{2}\log2\pi{e}\mathrm{Var}(\tilde{Z}_k)
\end{equation}
Next, we shall bound from above the rest of the terms on the RHS of
\eqref{eq:T3} in terms of $P_1$, $P_2$, $a_1$,
$\mathrm{Var}(\tilde{Z}_1)$, $a_2$ and $\mathrm{Var}(\tilde{Z}_2)$.

First, let us consider $h(Y_1^n)$ and $h(Y_2^n)$. We have
\begin{eqnarray}
h(Y_1^n) & \leq & \sum_{i=1}^{n}h(Y_{1i})\\
& \leq & \sum_{i=1}^{n}\frac{1}{2}\log2\pi
e\mathrm{Var}(Y_{1i})\\
& \leq &\frac{n}{2}\log\left[\frac{2\pi
e}{n}\sum_{i=1}^{n}\mathrm{Var}(Y_{1i})\right]\label{eq:T4}
\end{eqnarray}
where \eqref{eq:T4} is due to the concavity of the $\log$ function.
The variance
\begin{eqnarray}
\hspace{-20pt}
\mathrm{Var}(Y_{1i}) &=& \mathrm{Var}(X_{1i}+cX_{2i}+Z_{1i})\\
&=& \mathrm{Var}(X_{1i}+cX_{2i})+\mathrm{Var}(Z_{1i})\\
&=& \mathrm{Var}(X_{1i})+c^2\mathrm{Var}(X_{2i})+
2c\mathbb{E}[(X_{1i}-\mathbb{E}[X_{1i}])(X_{2i}-\mathbb{E}[X_{2i}])]+\mathrm{Var}(Z_{1i})
\label{eq:T5}
\end{eqnarray}
where the cross term
\begin{eqnarray}
&& \mathbb{E}[(X_{1i}-\mathbb{E}[X_{1i}])(X_{2i}-\mathbb{E}[X_{2i}])]\nonumber\\
&& \hspace{20pt}=\;
\mathbb{E}\left[\mathbb{E}\left[(X_{1i}-\mathbb{E}[X_{1i}])(X_{2i}-\mathbb{E}[X_{2i}])|W_0\right]\right]\\
&& \hspace{20pt}=\;
\mathbb{E}\left[\mathbb{E}\left[X_{1i}-\mathbb{E}[X_{1i}]|W_0\right]
\mathbb{E}\left[X_{2i}-\mathbb{E}[X_{2i}]|W_0\right]\right]\label{eq:T6}\\
&& \hspace{20pt}=\;
\mathbb{E}\left[\left(\mathbb{E}[X_{1i}|W_0]-\mathbb{E}[X_{1i}]\right)
\left(\mathbb{E}[X_{2i}|W_0]-\mathbb{E}[X_{2i}]\right)\right]\\
&& \hspace{20pt}\leq\;
\sqrt{\mathbb{E}\left[\left(\mathbb{E}[X_{1i}|W_0]-\mathbb{E}[X_{1i}]\right)^2\right]
\mathbb{E}\left[\left(\mathbb{E}[X_{2i}|W_0]-\mathbb{E}[X_{2i}]\right)^2\right]}\label{eq:T7}
\end{eqnarray}
where \eqref{eq:T6} follows from the independence of $X_{1i}$ and
$X_{2i}$ given $W_0$, and \eqref{eq:T7} follows from the
Cauchy-Schwartz inequality. Furthermore,
\begin{eqnarray}
\hspace{-20pt} E\left[\left(E[X_{ki}|W_0]-E[X_{ki}]\right)^2\right]
&=&
E\left[\left(E[X_{ki}|W_0]\right)^2\right]-\left(E[X_{ki}]\right)^2\\
&=&\left(E[X_{ki}^2]-\left(E[X_{ki}]\right)^2\right)-
\left(E[X_{ki}^2]-E\left[\left(E[X_{ki}|W_0]\right)^2\right]\right)\\
&=&\mathrm{Var}(X_{ki})-\mathrm{Var}(X_{ki}|W_0)\label{eq:T8}
\end{eqnarray}
Substituting \eqref{eq:T7} and \eqref{eq:T8} into \eqref{eq:T5}, we
may obtain
\begin{eqnarray}
\hspace{-20pt}&&\frac{1}{n}\sum_{i=1}^n\mathrm{Var}(Y_{1i})\nonumber\\
\hspace{-20pt}&& \hspace{5pt} \leq \;
\frac{1}{n}\sum_{i=1}^n\mathrm{Var}(X_{1i})+\frac{c^2}{n}\sum_{i=1}^n\mathrm{Var}(X_{2i})+\nonumber\\
\hspace{-20pt}&& \hspace{20pt}
\frac{2c}{n}\sum_{i=1}^n\sqrt{\left(\mathrm{Var}(X_{1i})-\mathrm{Var}(X_{1i}|W_0)\right)
\left(\mathrm{Var}(X_{2i})-\mathrm{Var}(X_{2i}|W_0)\right)}+\mathrm{Var}(Z_{1})\\
\hspace{-20pt}&& \hspace{5pt} \leq \;
\frac{1}{n}\sum_{i=1}^n\mathrm{Var}(X_{1i})+\frac{c^2}{n}\sum_{i=1}^n\mathrm{Var}(X_{2i})+\nonumber\\
\hspace{-20pt}&& \hspace{20pt}
2c\sqrt{\left(\frac{1}{n}\sum_{i=1}^n\mathrm{Var}(X_{1i})-\frac{1}{n}\sum_{i=1}^n\mathrm{Var}(X_{1i}|W_0)\right)
\left(\frac{1}{n}\sum_{i=1}^n\mathrm{Var}(X_{2i})-\frac{1}{n}\sum_{i=1}^n\mathrm{Var}(X_{2i}|W_0)\right)}+\nonumber\\
\hspace{-20pt}&& \hspace{20pt} \mathrm{Var}(Z_{1})\label{eq:T9}\\
\hspace{-20pt}&& \hspace{5pt} \leq \;
P+c^2P+2c\sqrt{(P-P_1)(P-P_2)}+1 \label{eqn:temp6}
\end{eqnarray}
where \eqref{eq:T9} follows from the fact that
\begin{equation}
g(x_1,x_2,y_1,y_2)=\sqrt{(y_1-x_1)(y_2-x_2)}
\end{equation}
is jointly concave for $x_1 \leq y_1$, $x_2 \leq y_2$. Hence, we
have
\begin{equation}
h(Y_1^n) \leq \frac{n}{2}\log\left[2\pi
e\left(P+c^2P+2c\sqrt{(P-P_1)(P-P_2)}+1\right)\right]
\end{equation}
and similarly
\begin{equation}
h(Y_2^n) \leq \frac{n}{2}\log\left[2\pi
e\left(P+c^2P+2c\sqrt{(P-P_1)(P-P_2)}+1\right)\right]
\end{equation}

Next, we consider $h(Y_1^n|U_1^n,W_0)$ and $h(Y_2^n|U_2^n,W_0)$. We
have
\begin{eqnarray}
&& h(Y_1^n|U_1^n,W_0)\nonumber\\
&& \hspace{5pt} \leq \;
\sum_{i=1}^{n}h(Y_{1i}|U_{1i},W_0)\\
&& \hspace{5pt} = \; \sum_{i=1}^{n}h(X_{1i}+cX_{2i}+Z_{1i}|cX_{1i}+\tilde{Z}_{1i},W_0)\\
&& \hspace{5pt} \leq \; \sum_{i=1}^{n}\frac{1}{2}\log\left[2\pi{e}\mathrm{Var}(X_{1i}+cX_{2i}+Z_{1i}|cX_{1i}+\tilde{Z}_{1i},W_0)\right]\\
&& \hspace{5pt} \leq \; \frac{n}{2}\log\left[\frac{2\pi{e}}{n}\sum_{i=1}^{n}\mathrm{Var}(X_{1i}+cX_{2i}+Z_{1i}|cX_{1i}+\tilde{Z}_{1i},W_0)\right]\\
&& \hspace{5pt} \le \;
\frac{n}{2}\log\left[\frac{2\pi{e}}{n}\sum_{i=1}^{n}
\left(\mathrm{Var}(X_{1i}|W_0)+c^2\mathrm{Var}(X_{2i}|W_0)+1-\right.\right.\nonumber\\
&& \hspace{20pt}
\left.\left.\frac{\left(c\mathrm{Var}(X_{1i}|W_0)+a_1\sqrt{\mathrm{Var}(\tilde{Z}_1)}\right)^2}
{c^2\mathrm{Var}(X_{1i}|W_0)+\mathrm{Var}(\tilde{Z}_1)}\right)\right]\\
&& \hspace{5pt} \leq \; \frac{n}{2}\log\left[2\pi{e}
\left(\frac{1}{n}\sum_{i=1}^{n}\mathrm{Var}(X_{1i}|W_0)+\frac{c^2}{n}\sum_{i=1}^{n}\mathrm{Var}(X_{2i}|W_0)+1-\right.\right.\nonumber\\
&& \hspace{20pt}
\left.\left.\frac{\left(\frac{c}{n}\sum_{i=1}^{n}\mathrm{Var}(X_{1i}|W_0)+a_1\sqrt{\mathrm{Var}(\tilde{Z}_1)}\right)^2}
{\frac{c^2}{n}\sum_{i=1}^{n}\mathrm{Var}(X_{1i}|W_0)+\mathrm{Var}(\tilde{Z}_1)}\right)\right]\label{eq:T10}\\
&& \hspace{5pt} = \; \frac{n}{2}\log\left[2\pi{e}
\left(P_1+c^2P_2+1-\frac{\left(cP_1+a_1\sqrt{\mathrm{Var}(\tilde{Z}_1)}\right)^2}
{c^2P_1+\mathrm{Var}(\tilde{Z}_1)}\right)\right]  \label{eqn:temp7}
\end{eqnarray}
where \eqref{eq:T10} follows from the fact that
\begin{equation}
g(x,y) :=
x+c^2y+1-\frac{\left(cx+a_1\mathrm{Var}(\tilde{Z}_1)\right)^2}{c^2x+\mathrm{Var}(\tilde{Z}_1)}
\end{equation}
is jointly concave for $x\geq0$, $y\geq0$. Similarly, we may also
obtain that
\begin{equation}
h(Y_2^n|U_2^n,W_0) \leq \frac{n}{2}\log\left[2\pi{e}
\left(P_2+c^2P_1+1-\frac{\left(cP_2+a_2\sqrt{\mathrm{Var}(\tilde{Z}_2)}\right)^2}
{c^2P_2+\mathrm{Var}(\tilde{Z}_2)}\right)\right]  \label{eqn:temp8}
\end{equation}

Finally, let us consider
$h(U_1^n|W_0)-h(Y_2^n|X_2^n,U_2^n,W_0)-h(Y_2^n|X_2^n,W_0)$ and
$h(U_2^n|W_0)-h(Y_1^n|X_1^n,U_1^n,W_0)-h(Y_1^n|X_1^n,W_0)$. We have
\begin{eqnarray}
&&
h(U_1^n|W_0)-h(Y_2^n|X_2^n,U_2^n,W_0)-h(Y_2^n|X_2^n,W_0)\nonumber\\
&& \hspace{20pt} =\;
h(cX_1^n+\tilde{Z}_1^n|W_0)-h(X_2^n+cX_1^n+Z_2^n|X_2^n,cX_2^n+\tilde{Z}_2^n,W_0)-\nonumber\\
&& \hspace{35pt} h(X_2^n+cX_1^n+Z_2^n|X_2^n,W_0)\\
&& \hspace{20pt} =\;
h(cX_1^n+\tilde{Z}_1^n|W_0)-h(cX_1^n+N_2^n|W_0)-h(cX_1^n+Z_2^n|W_0)
\label{eq:T11}
\end{eqnarray}
Assuming that
\begin{equation}
\mathrm{Var}(\tilde{Z}_1) \leq \mathrm{Var}(N_2)=1-a_2^2 \label{eqn:temp1}
\end{equation}
by the (conditional) entropy-power inequality, we have
\begin{equation}
h(cX_1^n+N_2^n|W_0) \geq
\frac{n}{2}\log\left(e^{\frac{2}{n}h(cX_1^n+\tilde{Z}_1^n|W_0)}+
2\pi{e}\left(1-a_2^2-\mathrm{Var}(\tilde{Z}_1)\right)\right)
\label{eq:T12}
\end{equation}
and
\begin{equation}
h(cX_1^n+Z_2^n|W_0) \geq
\frac{n}{2}\log\left(e^{\frac{2}{n}h(cX_1^n+\tilde{Z}_1^n|W_0)}+
2\pi{e}\left(1-\mathrm{Var}(\tilde{Z}_1)\right)\right)
\label{eq:T13}
\end{equation}
Substituting \eqref{eq:T12} and \eqref{eq:T13} into \eqref{eq:T11},
we may obtain
$$h(U_1^n|W_0)-h(Y_2^n|X_2^n,U_2^n,W_0)-h(Y_2^n|X_2^n,W_0) \leq
ng(t)$$ where
\begin{equation}
g(t):=t-\frac{1}{2}\log\left(e^{2t}+2\pi{e}\left(1-a_2^2-\mathrm{Var}(\tilde{Z}_1)\right)\right)-
\frac{1}{2}\log\left(e^{2t}+2\pi{e}\left(1-\mathrm{Var}(\tilde{Z}_1)\right)\right)
\end{equation}
and
\begin{eqnarray}
t &:=& \frac{1}{n}h(cX_1^n+\tilde{Z}_1^n|W_0)\\
& \leq & \frac{1}{n}\sum_{i=1}^nh(cX_{1i}+\tilde{Z}_{1i}|W_0)\\
& \leq & \frac{1}{2n}\sum_{i=1}^n\log\left[2\pi{e}\mathrm{Var}(cX_{1i}+\tilde{Z}_{1i}|W_0)\right]\\
& = & \frac{1}{2n}\sum_{i=1}^n\log\left[2\pi{e}\left(c^2\mathrm{Var}(X_{1i}|W_0)+\mathrm{Var}(\tilde{Z}_1)\right)\right]\\
& \leq & \frac{1}{2}\log\left[2\pi{e}\left(\frac{c^2}{n}\sum_{i=1}^n\mathrm{Var}(X_{1i}|W_0)+\mathrm{Var}(\tilde{Z}_1)\right)\right]\\
& = &
\frac{1}{2}\log\left[2\pi{e}\left(c^2P_1+\mathrm{Var}(\tilde{Z}_1)\right)\right]
\label{eq:T14}
\end{eqnarray}
The derivative
$$g'(t)=\frac{(2\pi{e})^2\left(1-a_2^2-\mathrm{Var}(\tilde{Z})\right)\left(1-\mathrm{Var}(\tilde{Z})\right)-e^{4t}}
{\left(e^{2t}+2\pi{e}(1-a_2^2-\mathrm{Var}(\tilde{Z}))\right)\left(e^{2t}+2\pi{e}(1-\mathrm{Var}(\tilde{Z}))\right)}$$
so $g(t)$ is a monotone increasing function for
\begin{equation}
t \leq
\frac{1}{4}\log\left[(2\pi{e})^2\left(1-a_2^2-\mathrm{Var}(\tilde{Z}_1)\right)\left(1-\mathrm{Var}(\tilde{Z}_1)\right)\right]
\end{equation}
Assuming that
\begin{equation}
c^2P_1 \leq
\sqrt{(1-a_2^2-\mathrm{Var}(\tilde{Z}_1))(1-\mathrm{Var}(\tilde{Z}_1))}-\mathrm{Var}(\tilde{Z}_1)\label{eqn:temp2}
\end{equation}
we have from \eqref{eq:T14}
\begin{eqnarray}
\hspace{-20pt}&&
h(U_1^n|W_0)-h(Y_2^n|X_2^n,U_2^n,W_0)-h(Y_2^n|X_2^n,W_0)\nonumber\\
\hspace{-20pt}&& \hspace{20pt} \leq \;
g\left(\frac{1}{2}\log\left[2\pi{e}\left(c^2P_1+\mathrm{Var}(\tilde{Z}_1)\right)\right]\right)\\
\hspace{-20pt}&& \hspace{20pt} \leq \; \frac{n}{2}
\left[\log2\pi{e}\left(c^2P_1+\mathrm{Var}(\tilde{Z}_1)\right)
-\log2\pi{e}\left(c^2P_1+1-a_2^2\right)-
\log2\pi{e}\left(c^2P_1+1\right)\right]  \label{eqn:temp9}
\end{eqnarray}
Similarly, assuming that
\begin{equation}
\mathrm{Var}(\tilde{Z}_2) \leq 1-a_1^2 \label{eqn:temp3}
\end{equation}
and
\begin{equation}
c^2P_2 \leq
\sqrt{\left(1-a_1^2-\mathrm{Var}(\tilde{Z}_2)\right)\left(1-\mathrm{Var}(\tilde{Z}_2)\right)}-\mathrm{Var}(\tilde{Z}_2) \label{eqn:temp4}
\end{equation}
we have
\begin{eqnarray}
\hspace{-20pt} &&
h(U_2^n|W_0)-h(Y_1^n|X_1^n,U_1^n,W_0)-h(Y_1^n|X_1^n,W_0)\nonumber\\
\hspace{-20pt} && \hspace{20pt} \leq \; \frac{n}{2}
\left[\log2\pi{e}\left(c^2P_2+\mathrm{Var}(\tilde{Z}_2)\right)
-\log2\pi{e}\left(c^2P_2+1-a_1^2\right)-
\log2\pi{e}\left(c^2P_2+1\right)\right]. \label{eqn:temp10}
\end{eqnarray}

Therefore, we have
\begin{align}
  R_0+R_1+R_2 \le \max_{0 \leq P_1, P_2 \le P} g(P_1,P_2).
\end{align}

Next we argue that we only need to consider the case when $P_1 = P_2$.

Recall that  $P_k (k=1,2)$ is defined as
 \begin{align}
     P_k := \frac{1}{n}\sum_{i=1}^{n}\text{Var}(X_{ki}|W_0).
 \end{align}

 We show that in this symmetric model, given any transmission scheme, one can easily construct another transmission scheme achieving the same sum rate with $P_1 = P_2$. Indeed, suppose in the given transmission scheme $P_1 \neq P_2$. We construct another transmission scheme as follows:
 \begin{itemize}
    \item In the first time block, we use the same code book of the given transmission scheme.
    \item In the second time block, since the channel is symmetric, we can switch the roles of user 1 and user 2 and use the same transmission scheme achieving the same sum rate.
  \end{itemize}
Hence the new transmission scheme achieves the same sum rate with $P_1' = P_2' = \frac{P_1 + P_2}{2}$.

Therefore, we have
\begin{align}
  R_0+R_1+R_2 \le \max_{0 \leq P_1 = P_2 \le P} g(P_1,P_2).
\end{align}

This completes the proof of Theorem \ref{thm:upper}.

\end{IEEEproof}

\section{Tightness of Upper Bound in the Low Interference Regime}\label{sec:tightness}
In this section, we first given a lower bound on the sum capacity, and then show that the upper bound given in Theorem \ref{thm:upper} matches the lower bound in the low interference regime. 


A simple coding scheme is that each transmitter splits the power $P$ into two parts, one for common message $M_0$ and one for the privacy message, and does channel coding for each message independently, and each receiver decodes the intended messages by using successive interference cancellation. The transmission sum rates of this superposition coding scheme are a natural lower bound for the sum capacity.

\begin{lemma}\label{Inner}
Given $P_1 = P_2 = P-P_0$, the maximum sum rates achieved by the above superposition coding scheme is
\begin{align}
 &  R_0 + R_1 + R_2 \nonumber \\
 \le&  \frac{1}{2} \left(\log (1 + \frac{P_1 + (1+c)^2P_0}{c^2P_2 + 1}) + \log(1 + \frac{P_2}{c^2P_1 + 1})\right)  \label{eq:inner}\\
 =& I(X_{1G},X_{0G};Y_{1G}) + I(X_{2G},Y_{2G}|X_{0G}) \nonumber \\
 =& I(X_{2G},X_{0G};Y_{2G}) + I(X_{1G},Y_{1G}|X_{0G}), \nonumber
\end{align}
where
\begin{align*}
  X_{1G} &= X_{11G} + \sqrt{P - P1} X_{0G}, \\
  X_{2G} &= X_{22G} + \sqrt{P - P2} X_{0G}, \\
  Y_{1G} &= X_{11G} + (1+c)\sqrt{P - P1}X_{0G} + c X_{22G} + Z_1, \\
  Y_{2G} &= X_{22G} + (1+c)\sqrt{P - P2}X_{0G} + c X_{11G} + Z_2,
\end{align*}
and $X_{11G}, X_{22G}, X_{0G}$ are zero mean Gaussian random variables with variances $P_1, P_2$ and 1.

\end{lemma}

\begin{IEEEproof}
Gaussian random variables $X_{11G}, X_{22G}$ and $X_{0G}$ correspond to the codebooks for messages $W_1, W_2$ and $W_0$ in the superposition coding scheme.

The channel input-output relation is
\begin{align*}
  Y_{1G} &= X_{11G} + (1+c)\sqrt{P - P1} X_{0G} + c X_{22G} + Z_1 \\
  Y_{2G} &= X_{22G} + (1+c)\sqrt{P - P2} X_{0G} + c X_{11G} + Z_2,
\end{align*}

Treating interference as noise, we can write down the expressions of the MAC capacity region (six inequalities in total) and get the maximum sum rates \eqref{eq:inner} by using Fourier-Motzkin elimination. More specifically, given $a,b,c,d,e$ and $f$, from the following inequalities
\begin{eqnarray}
  R_0 &\le& a \\
  R_1 &\le& b \\
  R_0 + R_1 &\le& c \\
  R_0 &\le& d \\
  R_2 &\le& e \\
  R_0 + R_2 &\le& f
\end{eqnarray}
we can get a tight upper bound for $R_0 + R_1 + R_2$ via Fourier-Motzkin elimination:
\begin{equation}
    R_0 + R_1 + R_2 \le \min\{a+b+e, b+d+e, c+e, b+f \}.
\end{equation}
In our MAC channel, we have the implicit conditions $a + b \ge c$, $d + e \ge f$ and by symmetry $a=d,b=e,c=f$. Thus $\min\{a+b+e, b+d+e, c+e, b+f \}$ = $b+f$ = $c+e$. This completes our proof.
\end{IEEEproof}

We show that the lower bound \eqref{eq:inner} and the upper bound \eqref{thm:upper} match for all $0 \le P_1 = P_2 \le P$ if the parameters $P$ and $c$ satisfy certain conditions. More precisely,

\begin{theorem}\label{thm:main}
Given $P$ and $c$, if there exist nonnegative parameters $a$ and $b$ such that
 \begin{eqnarray}
        c(1+c^2P) &=& ab   \label{eq:1}\\
        c^2P &\le& \sqrt{(1-a^2-b^2)(1-b^2)} - b^2 \label{eq:2}\\
        a^2 + b^2 &\le& 1 \label{eq:3}.
 \end{eqnarray}
then the lower bound \eqref{eq:inner} and the outer bound \eqref{eqn:upper} match, i.e.,
\begin{align}
  g(P_1,P_2) = \frac{1}{2} \left(\log (1 + \frac{P_1 + (1+c)^2P_0}{c^2P_2 + 1}) + \log(1 + \frac{P_2}{c^2P_1 + 1})\right)
\end{align}
for $0 \le P_1 = P_2 \le P$.
\end{theorem}


Before proving Theorem \ref{thm:main}, note that by setting $a^2 = \frac{1}{2}$ to be a valid solution of the above constraints, we can derive a sufficient condition under which the lower bound \eqref{eq:inner} and the outer bound \eqref{thm:upper} match.
\begin{corollary}
If the parameters $P$ and $c$ satisfy the following low interference conditions:
\begin{equation}\label{eq:4}
 c^4 P^2 + (4c^2P+3)c^2(1+c^2P)^2 \le \frac{1}{2}
\end{equation}
and
\begin{equation}\label{eq:5}
 c(1+c^2P) \le \frac{1}{2},
\end{equation}
then the lower bound \eqref{eq:inner} and the outer bound \eqref{thm:upper} match.
\end{corollary}

\begin{IEEEproof}
The idea is very simple. We find the region in which $a^2 = \frac{1}{2} $ is a valid solution of (\ref{eq:1}), (\ref{eq:2}) and (\ref{eq:3}).

Consider (\ref{eq:1}) and (\ref{eq:3}) and use the familiar inequality $a^2 + b^2 \ge 2ab$, we have
\begin{equation}
    c(1+c^2P) = ab  \le \frac{a^2 + b^2}{2} \le \frac{1}{2}.
\end{equation}

Let $m := c(1+c^2P)$, so $ m \le \frac{1}{2}$. From (\ref{eq:1}) and (\ref{eq:3}) we get
\begin{equation}
a^2 + \frac{m^2}{a^2} \le 1,
\end{equation}
thus
\begin{equation}
    \frac{1-\sqrt{1-4m^2}}{2} \le a^2 \le \frac{1+\sqrt{1-4m^2}}{2}.
\end{equation}

In (\ref{eq:2}), let $a^2 = \frac{1}{2}$, then we get
\begin{equation}
    c^2P \le \sqrt{(1-\frac{1}{2}-2m^2)(1-2m^2)} - 2m^2,
\end{equation}
equivalently,
\begin{equation}
   c^4 P^2 + (4c^2P+3)c^2(1+c^2P)^2 \le \frac{1}{2}.
\end{equation}

So if $P$ and $c$ satisfy (\ref{eq:4}) and (\ref{eq:5}), then $a^2 = \frac{1}{2}$ and $b^2 = 2 m^2$ are valid solutions of (\ref{eq:1}), (\ref{eq:2}) and (\ref{eq:3}). Therefore, the lower bound \eqref{eq:inner} and the outer bound \eqref{thm:upper} match due to Theorem \ref{thm:main}.

\end{IEEEproof}


\subsection{Proof of Theorem \ref{thm:main}}

The following Lemma \ref{Outer} is a restatement of \eqref{eqn:temp1}, \eqref{eqn:temp2}, \eqref{eqn:temp3}, \eqref{eqn:temp4} in Theorem \ref{thm:upper}, which is  the  so-called ``useful genie condition" defined in \cite{Sreekanth09}. Lemma \ref{Outer} says that if the useful genie condition is satisfied, then the capacity of genie aided channel is achieved by channel input with Gaussian distributions.

To simplify the notation, define $b_k := \sqrt{\text{Var}(\tilde{Z_k})}$.

\begin{lemma}\label{Outer}
Given the conditional variances $P_1$ and $P_2$, if
\begin{eqnarray}
  b_2^2 &\le& 1 - a_1^2  \\
  c^2P_2 &\le& \sqrt{(1-a_1^2-b_2^2)(1-b_2^2)} - b_2^2 \\
  b_1^2 &\le& 1 - a_2^2 \\
  c^2P_1 &\le& \sqrt{(1-a_2^2-b_1^2)(1-b_1^2)} - b_1^2
\end{eqnarray}

then
\begin{eqnarray}
  (R_0 + R_1 + R_2 - \epsilon)  &\le& \frac{1}{2} ( I(X_{1G};Y_{1G},U_{1G}|X_{0G}) + I(Y_{1G};X_{1G},X_{0G}) \nonumber\\
   & & +   I(X_{2G};Y_{2G},U_{2G}|X_{0G}) + I(Y_{2G};X_{2G},X_{0G}) ),
\end{eqnarray}
where $X_{kG}$ are the zero mean Gaussian random variables, and $U_{kG}, Y_{kG}$ are the corresponding Gaussian genie and output. More specifically,
\begin{eqnarray}
  X_{1G} &:=& X_{11G} + \sqrt{P - P1} X_{0G} , \\
  X_{2G} &:=& X_{22G} + \sqrt{P - P2} X_{0G} ,
\end{eqnarray}
where $X_{11G}, X_{22G}$ and $X_{0G}$ are independent zero mean Gaussian random variables with variance $P_1, P_2$ and $1$, respectively.
And accordingly,
\begin{eqnarray}
  Y_{1G} &=& X_{1G} + c X_{2G} + Z_1\\
  Y_{2G} &=& X_{2G} + c X_{1G} + Z_2\\
  U_{1G} &=& c X_{1G} + \tilde{Z}_1 \\
  U_{2G} &=& c X_{2G} + \tilde{Z}_2
\end{eqnarray}
\end{lemma}

\begin{proof}
In the proof of Theorem \ref{thm:upper}, \eqref{eqn:temp5} is
\begin{eqnarray*}
  n (R_0 + R_1 + R_2 - \epsilon) &\le& \frac{1}{2} (h(Y_1^n) + I(X_1^n;Y_1^n,U_1^n|W_0) - h(Y_1^n|X_1^n,W_0) \\
   & & + h(Y_2^n) + I(X_2^n;Y_2^n,U_2^n|W_0) - h(Y_2^n|X_2^n,W_0) ).
\end{eqnarray*}
The RHS is exactly
\begin{equation}
  \frac{1}{2}(I(X_1^n;Y_1^n,U_1^n|W_0) + I(Y_1^n;X_1^n,W_0) + I(X_2^n;Y_2^n,U_2^n|W_0) + I(Y_2^n;X_2^n,W_0)),
\end{equation}
and for each term we have derived an outer bound (c.f. \eqref{eq:T3},  \eqref{eqn:temp6}, \eqref{eqn:temp7},  \eqref{eqn:temp8}, \eqref{eqn:temp9}, \eqref{eqn:temp9}) in the proof of Theorem \ref{thm:upper}. It is easy to verify that the derived bound for each term can be obtained by replacing every term in the mutual information and entropy expressions by the corresponding Gaussian random variables. In this way, we can equivalently write the function $f(P_1,P_2,a_1^2,\text{Var}(\tilde{Z}_1),a_2^2,\text{Var}(\tilde{Z}_2))$ as
\begin{equation}
    \frac{1}{2} ( I(X_{1G};Y_{1G},U_{1G}|X_{0G}) + I(Y_{1G};X_{1G},X_{0G}) +   I(X_{2G};Y_{2G},U_{2G}|X_{0G}) + I(Y_{2G};X_{2G},X_{0G}) ).
\end{equation}
This completes the proof.
\end{proof}

As stated before, Lemma \ref{Outer} is just a restatement of the upper bound in Theorem \ref{thm:upper} in terms of mutual information among Gaussian random variables. The advantage of doing so is that we can easily compare the lower bound and outer bound, and study under what conditions they match.


The following lemma deals with the so-called smart genie condition defined  in \cite{Sreekanth09}. Under this condition, the genie-aided channel sum capacity is same as the one achieved by superposition coding and successive interference cancellation in the genie-free channel.

\begin{lemma}
Given fixed conditional variances $P_1 = P_2$, if there exist parameters $a$ and $b$ satisfying
\begin{equation}
    c(1 + c^2P_1) = ab
\end{equation}
and the useful genie conditions in Lemma \ref{Outer}, then the sum capacity $g(P_1,P_2)$ of the genie aided channel is same as the one achieved by superposition coding and successive interference cancellation in the genie-free channel.
\end{lemma}

\begin{proof}
We emphasize that here the inner bound and outer bound are bounds for specific given $P_1$ and $P_2$, where $0 \le P_1 = P_2 \le P$.

By Lemma \ref{Outer} and Lemma \ref{Inner}, given $P_1 = P_2$, the gap between outer bound and inner bound is
\begin{equation}
    \frac{1}{2} ( I(X_{1G};U_{1G}|X_{0G},Y_{1G}) + I(X_{2G};U_{2G}|X_{0G},Y_{2G})).
\end{equation}
If the gap is zero, i.e., outer bound and inner bound match, each term must be zero since mutual information is nonnegative.
Indeed, by symmetry we have
\begin{equation}
    I(X_{1G};U_{1G}|X_{0G},Y_{1G}) =  I(X_{2G};U_{2G}|X_{0G},Y_{2G}).
\end{equation}

Recall that
\begin{eqnarray}
  X_{1G} &:=& X_{11G} + \sqrt{P - P_1} X_{0G} , \\
  X_{2G} &:=& X_{22G} + \sqrt{P - P_2} X_{0G} ,
\end{eqnarray}
where $X_{11G}, X_{22G}$ and $X_{0G}$ are independent zero mean Gaussian random variables with variance $P_1, P_2$ and $1$, respectively.

By defining
\begin{equation}
    \tilde{X}_{kG} := X_{kG} - E[X_{kG}|X_{0G}] = X_{kkG},
\end{equation}
we have
\begin{eqnarray*}
   & & I(X_{1G};U_{1G}|X_{0G},Y_{1G}) = 0 \\
  &\Leftrightarrow& I(X_{1G};cX_{1G}+\tilde{Z}_1|X_{0G},X_{1G}+cX_{2G}+Z_1) = 0 \\
  &\Leftrightarrow& I(\tilde{X}_{1G};c\tilde{X}_{1G}+\tilde{Z}_1|X_{0G},\tilde{X}_{1G}+c\tilde{X}_{2G}+Z_1) = 0 \\
  &\Leftrightarrow& I(\tilde{X}_{1G};c\tilde{X}_{1G}+\tilde{Z}_1|\tilde{X}_{1G}+c\tilde{X}_{2G}+Z_1) = 0 ,
\end{eqnarray*}
where in the last step we use the fact that $X_{0G}$ is independent of $\tilde{X}_{kG}$, i.e., $X_{0G}$ is independent of  $X_{kkG}$.

The last condition is equivalent to the Markov Chain condition:
\begin{equation}
    \tilde{X}_{1G} \rightarrow \tilde{X}_{1G} + c\tilde{X}_{2G} + Z_1 \rightarrow c\tilde{X}_{1G} + \tilde{Z}_1.
\end{equation}

Since all the random variables are Gaussian, by the fact that Gaussian random variables $X \rightarrow Y \rightarrow Z$ if and only if
\begin{equation}
    \text{Cov}(X,Z) = \text{Cov}(X,Y)  \text{Cov}(Y)^{-1}  \text{Cov}(Y,Z),
\end{equation}
we get the smart genie condition
\begin{equation}
    c(1+c^2P_1) = ab.
\end{equation}

\end{proof}

So far, we have shown that given $P_1 = P_2$, under what conditions inner bound and outer bound match. The next step is to show for all $P_1$ and $P_2$, where $ 0 \le P_1 = P_2 \le P$, there exist parameters $a(P_1)$ and $b(P_1)$ satisfying both useful genie and smart genie conditions, and this will conclude our proof.

More specifically, we want to show that in some low interference regime, for any $P_1 \le P$, there exist nonnegative parameters $a$ and $b$ (we emphasize here $a$ and $b$ can be a function of $P_1$) satisfying the following conditions
 \begin{eqnarray}
        c(1+c^2P_1) &=& a b   \\
        c^2P_1 &\le& \sqrt{(1-a^2-b^2)(1-b^2)} - b^2 \\
        a^2 + b^2 &\le& 1 .
 \end{eqnarray}

It is easy to see that we only need to consider the case that $P_1 = P$ for the above constraints, since if there exist $a$ and $b$ satisfying the conditions for the case $P_1 = P$, which are exactly equations (\ref{eq:1}), (\ref{eq:2}) and (\ref{eq:3}),            then it has solutions for all $P_1 \le P$ (as $P_1$ decreases, we can fix $a$ and decrease the value of $b$ to satisfy all the constraints).

This completes the proof of Theorem \ref{thm:upper}.

\section{Optima Common Message Rate in the Low Interference Regime}\label{sec:zerocommon}
In this section, we show that  in the low interference regime defined in \eqref{eq:4} and \eqref{eq:5}, the optimal power allocation is to set $P_0$ to be zero to achieve sum capacity.

\begin{lemma}\label{pzero}
If
\begin{equation}\label{mono}
    (c^4 + 2c^3 + c^2)P + c^2 + 2c - 1 \le 0,
\end{equation}
then the sum rate (\ref{eq:inner})
\begin{equation}
 R(P_0) := \frac{1}{2} (\log (1 + \frac{(P-P_0) + (1+c)^2P_0}{c^2(P-P_0) + 1}) + \log(1 + \frac{(P-P_0)}{c^2(P-P_0) + 1}))
\end{equation}
is a decreasing function of $P_0$ and thus is maximized by setting $P_0$ to be zero.
\end{lemma}

\begin{proof}
To simplify the notation, define
\begin{eqnarray}
  a &:=& P + \frac{1}{c^2} \\
  b &:=& P + \frac{1+P}{c^2} \\
  d &:=& \frac{2}{c} \\
  e &=& 1 + \frac{1}{c^2}.
\end{eqnarray}

So
\begin{equation}
  R(P_0) =  \frac{1}{2} (\log (\frac{b + dP_0}{a - P_0}) + \log\frac{(b-eP_0)}{a-P_0}) ) ,
\end{equation}
and the derivative of $R(P_0)$ is
\begin{equation}
    \frac{dR}{dP_0} = \frac{1}{2} \frac{(bd-2aed-be)P_0 + abd + 2b^2 - abe}{(dP_0 + b)(P_0-a)(eP_0-b)}.
\end{equation}

In total there are three poles and one zero, which are
\begin{eqnarray*}
  p_1 &=& -\frac{b}{d} = - (\frac{c P}{2} + \frac{1 + P}{2c}) <0 , \\
  p_2 &=& a = P + \frac{1}{c^2} >0 ,\\
  p_3 &=& \frac{b}{e} = P + \frac{1}{1 + c^2}  <p_2 \\
  z_1 &=& \frac{(P+\frac{1+P}{c^2})[(1+\frac{2}{c}+\frac{1}{c^2})P + \frac{1}{c^2} +\frac{2}{c^3}] - \frac{1}{c^4}}{ (1+\frac{2}{c}+\frac{2}{c^2}+\frac{2}{c^3}+\frac{1}{c^4})P + (\frac{1}{c^2}+\frac{2}{c^3}+\frac{1}{c^4}+\frac{4}{c^5})  }.
\end{eqnarray*}
It is easy to see
\begin{equation}
    p_1 < 0 < P < p_3 < p_2.
\end{equation}
Now we only need to consider the value of $z_1$. If $c$ is sufficiently small, then $z_1$ is negative, and thus $\frac{dR}{dP_0}$ is negative on $[\max\{p_1,z_1\}, t_3]$, so $R(P_0)$ is a  monotonically decreasing function on $[0, P]$, since $p_1,z_1 <0$ and $t_3 > P$.
The condition of $z_1 \le0$ is exactly the inequality (\ref{mono}).
\end{proof}

Lastly, we prove that the conditions (\ref{eq:4}) and (\ref{eq:5}) imply (\ref{mono}).

Denote by $\Gamma_A$ the region of $(c, P)$ determined by (\ref{eq:4}) and (\ref{eq:5}), and denote by $\Gamma_B$ the region of $(c, P)$ determined by (\ref{mono}).

\begin{theorem}\label{equal}
$\Gamma_A \subset \Gamma_B$.  Thus, when \eqref{eq:4} and \eqref{eq:5} hold, the optimal rate for the common message is zero to achieve the sum capacity.
\end{theorem}

\begin{proof}
Note that LHS of (\ref{eq:4}), (\ref{eq:5}) and (\ref{mono}) are  increasing functions of $P$. Fix $c$, and let $(c,P_A)$ and $(c, P_B)$ be the points on the boundary of $\Gamma_A$ and $\Gamma_B$, respectively. To show $\Gamma_A \subset \Gamma_B$, it is sufficient to show $ P_A \le P_B $. Equivalently, it is sufficient to show either
\begin{equation}
    c^4 P_B^2 + (4c^2P_B+3)c^2(1+c^2P_B)^2 \ge \frac{1}{2},
\end{equation}
or
\begin{equation}
    c(1+c^2P_B) \ge \frac{1}{2}.
\end{equation}

Indeed, we will prove
\begin{eqnarray}
  c^4 P_B^2 + (4c^2P_B+3)c^2(1+c^2P_B)^2 &\ge& \frac{1}{2}, \\
  c(1+c^2P_B) &\le& \frac{1}{2}.
\end{eqnarray}

First from (\ref{mono}) we get
\begin{equation}
    P_B = \frac{1 - 2c - c^2}{c^4 + 2c^3 + c^2}.
\end{equation}
Therefore,
\begin{eqnarray}
   & & c(1+c^2P_B) \le \frac{1}{2} \\
   &\Leftrightarrow& c(1 + \frac{1 - 2c - c^2}{c^2 + 2c +1}) \le \frac{1}{2} \\
   &\Leftrightarrow& \frac{2c}{c^2 + 2c +1}  \le \frac{1}{2}\\
   &\Leftrightarrow& 4c \le c^2 + 2c +1 \\
   &\Leftrightarrow& 0 \le c^2 - 2c +1 \\
   &\Leftrightarrow& 0 \le (c-1)^2.
\end{eqnarray}
The last step holds obviously, so $c(1+c^2P_B) \le \frac{1}{2}$.

Next we prove $c^4 P_B^2 + (4c^2P_B+3)c^2(1+c^2P_B)^2 \ge \frac{1}{2}$.
\begin{eqnarray}
   & & c^4 P_B^2 + (4c^2P_B+3)c^2(1+c^2P_B)^2 \ge \frac{1}{2} \\
   &\Leftrightarrow&  (\frac{1 - 2c - c^2}{c^2 + 2c +1})^2 + (4\frac{1 - 2c - c^2}{c^2 + 2c +1} + 3)c^2(1 + \frac{1 - 2c - c^2}{c^2 + 2c +1})^2  \ge \frac{1}{2}\\
   &\Leftrightarrow& (\frac{1 - 2c - c^2}{c^2 + 2c +1})^2 + \frac{-c^2 - 2c + 7}{c^2 + 2c +1}\frac{4c^2}{(c^2 + 2c +1)^2} \ge \frac{1}{2}\\
   &\Leftrightarrow& (1 - 2c - c^2)^2 (c^2 + 2c + 1) + 4c^2(-c^2 -2c + 7) \ge \frac{1}{2}(c+1)^6 \\
   &\Leftrightarrow&  (1 + 2c^2 + c^4 -4c + 4c^3)(c^2 + 2c +1) - 4c^4 - 8c^3 + 28c^2 \ge \frac{1}{2}(c+1)^6\\
   &\Leftrightarrow&  (-5c^2 + 11c^4 + c^6 + 4c^3 +6c^5 -2c +1) - 4c^4 -8c^3 + 28c^2  \ge \frac{1}{2}(c+1)^6\\
   &\Leftrightarrow&  c^6 + 6c^5 + 7c^4 - 4c^3 + 23c^2 - 2c + 1 \\
   && \ge \frac{1}{2}(c^6 + 6c^5 + 15c^4 + 20c^3 + 15c^2 + 6c +1)\\
   &\Leftrightarrow& c^6 + 6c^5 - c^4 - 28c^3 + 31c^2 - 10c + 1 \ge 0  \\
   &\Leftrightarrow& (c - 1)^2 (c^2 + 4c - 1)^2 \ge 0.
\end{eqnarray}
Since the last step holds, we have $c^4 P_B^2 + (4c^2P_B+3)c^2(1+c^2P_B)^2 \ge \frac{1}{2}$.
This completes the proof of Theorem \ref{equal}.
\end{proof}

\section{Conclusion}\label{sec:conclusion}
We consider symmetric two-user Gaussian interference channel with common messages. We derive an upper bound on the sum capacity, and show that the upper bound is tight in the low interference regime, where the optimal transmission scheme is to send no common messages and each receiver treats interference as noise. Our result shows that although the availability of common messages provides a  cooperation opportunity for transmitters, in the low interference regime the presence of common messages does not help increase the sum capacity.

\section*{Acknowledgment}\label{sec:ackowledgementd}

We thank Prof. Pramod Viswanath and Dr. Sreekanth Annapureddy for the helpful discussions.

Research of Quan Geng was supported in part by Prof. Pramod Viswanath's National Science Foundation grant No. CCF-1017430.

\bibliographystyle{IEEEtran}

\bibliography{reference}

\begin{thebibliography}{1}
\providecommand{\url}[1]{#1}
\csname url@samestyle\endcsname
\providecommand{\newblock}{\relax}
\providecommand{\bibinfo}[2]{#2}
\providecommand{\BIBentrySTDinterwordspacing}{\spaceskip=0pt\relax}
\providecommand{\BIBentryALTinterwordstretchfactor}{4}
\providecommand{\BIBentryALTinterwordspacing}{\spaceskip=\fontdimen2\font plus
\BIBentryALTinterwordstretchfactor\fontdimen3\font minus
  \fontdimen4\font\relax}
\providecommand{\BIBforeignlanguage}[2]{{%
\expandafter\ifx\csname l@#1\endcsname\relax
\typeout{** WARNING: IEEEtran.bst: No hyphenation pattern has been}%
\typeout{** loaded for the language `#1'. Using the pattern for}%
\typeout{** the default language instead.}%
\else
\language=\csname l@#1\endcsname
\fi
#2}}
\providecommand{\BIBdecl}{\relax}
\BIBdecl

\bibitem{Cover12}
T.~M. Cover and J.~A. Thomas, \emph{Elements of information theory}.\hskip 1em
  plus 0.5em minus 0.4em\relax John Wiley \& Sons, 2012.

\bibitem{Etkin08}
R.~Etkin, D.~Tse, and H.~Wang, ``Gaussian interference channel capacity to
  within one bit,'' \emph{IEEE Transactions on Information Theory}, vol.~54,
  no.~12, pp. 5534--5562, 2008.

\bibitem{Sato81}
H.~Sato, ``The capacity of the gaussian interference channel under strong
  interference (corresp.),'' \emph{IEEE Transactions on Information Theory},
  vol.~27, no.~6, pp. 786--788, 1981.

\bibitem{Sreekanth09}
V.~Annapureddy and V.~Veeravalli, ``Gaussian interference networks: Sum
  capacity in the low-interference regime and new outer bounds on the capacity
  region,'' \emph{IEEE Transactions on Information Theory}, vol.~55, no.~7, pp.
  3032--3050, 2009.

\bibitem{Shang09}
X.~Shang, G.~Kramer, and B.~Chen, ``A new outer bound and the
  noisy-interference sum rate capacity for gaussian interference channels,''
  \emph{IEEE Transactions on Information Theory}, vol.~55, no.~2, pp. 689--699,
  2009.

\bibitem{Motahari09}
A.~Motahari and A.~Khandani, ``Capacity bounds for the gaussian interference
  channel,'' \emph{IEEE Transactions on Information Theory}, vol.~55, no.~2,
  pp. 620--643, 2009.

\end{thebibliography}

\end{document}